\documentclass[10pt,draft]{amsart}
\usepackage{amsmath,amssymb,amsthm, color}

\begin{document}

\newtheorem{lem}{Lemma}[section]
\newtheorem{prop}{Proposition}[section]
\newtheorem{cor}{Corollary}[section]
\numberwithin{equation}{section}
\newtheorem{thm}{Theorem}[section]
\theoremstyle{remark}
\newtheorem{example}{Example}[section]
\newtheorem*{ack}{Acknowledgment}
\theoremstyle{definition}
\newtheorem{definition}{Definition}[section]
\theoremstyle{remark}
\newtheorem*{notation}{Notation}
\theoremstyle{remark}
\newtheorem{remark}{Remark}[section]
\newenvironment{Abstract}
{\begin{center}\textbf{\footnotesize{Abstract}}%
\end{center} \begin{quote}\begin{footnotesize}}
{\end{footnotesize}\end{quote}\bigskip}
\newenvironment{nome}
{\begin{center}\textbf{{}}%
\end{center} \begin{quote}\end{quote}\bigskip}

\newcommand{\triple}[1]{{|\!|\!|#1|\!|\!|}}
\newcommand{\xx}{\langle x\rangle}
\newcommand{\ep}{\varepsilon}
\newcommand{\al}{\alpha}
\newcommand{\be}{\beta}
\newcommand{\de}{\partial}
\newcommand{\la}{\lambda}
\newcommand{\La}{\Lambda}
\newcommand{\ga}{\gamma}
\newcommand{\del}{\delta}
\newcommand{\Del}{\Delta}
\newcommand{\sig}{\sigma}
\newcommand{\ome}{\Omega^n}
\newcommand{\Ome}{\Omega^n}
\newcommand{\C}{{\mathbb C}}
\newcommand{\N}{{\mathbb N}}
\newcommand{\Z}{{\mathbb Z}}
\newcommand{\R}{{\mathbb R}}
\newcommand{\T}{{\mathbb T}}
\newcommand{\Rn}{{\mathbb R}^{n}}
\newcommand{\Rnu}{{\mathbb R}^{n+1}_{+}}
\newcommand{\Cn}{{\mathbb C}^{n}}
\newcommand{\spt}{\,\mathrm{supp}\,}
\newcommand{\Lin}{\mathcal{L}}
\newcommand{\SSS}{\mathcal{S}}
\newcommand{\F}{\mathcal{F}}
\newcommand{\xxi}{\langle\xi\rangle}
\newcommand{\eei}{\langle\eta\rangle}
\newcommand{\xei}{\langle\xi-\eta\rangle}
\newcommand{\yy}{\langle y\rangle}
\newcommand{\dint}{\int\!\!\int}
\newcommand{\hatp}{\widehat\psi}
\renewcommand{\Re}{\;\mathrm{Re}\;}
\renewcommand{\Im}{\;\mathrm{Im}\;}

\title[Relativistic Hardy inequality]%
{Relativistic Hardy inequalities in magnetic fields}

\author{Luca Fanelli}
\address{Luca Fanelli: SAPIENZA Universit$\grave{\text{a}}$ di Roma, Dip. di Matematica "G. Castelnuovo". P.le A.
Moro 5, 00185, Roma, Italy}
\email{fanelli@mat.uniroma1.it}
\author{Luis Vega}
\address{Luis Vega: Universidad del Pa\'is Vasco, Dep. de Matem\'aticas. Apartado 644, 48080, Bilbao, Spain \& Basque Center for Applied Mathematics (BCAM), Alameda de Mazarredo 14, 48009, Bilbao, Spain}
\email{luis.vega@ehu.es -- lvega@bcamath.org}
\author{Nicola Visciglia}
\address{Nicola Visciglia: Universit$\grave{\text{a}}$ Degli Studi di Pisa, Dip. di Matematica "L. Tonelli".
Largo Bruno Pontecorvo 5 I - 56127 Pisa. Italy}
\email{viscigli@dm.unipi.it}


\begin{abstract}
We deal with Dirac operators with external homogeneous magnetic fields. Hardy-type inequalities related to these operators are investigated: for a suitable class of transversal magnetic fields, we prove a Hardy inequality with the same best constant as in the free case. This leaves naturally open an interesting question whether  there exist magnetic fields for which a Hardy inequality with a better constant than the usual one, in connection with the well known diamagnetic phenomenon arising in non-relativistic models.
\end{abstract}

\thanks{
The first and third authors were supported by the Italian project FIRB 2012: "Dispersive dynamics: Fourier Analysis and Variational Methods"}

\date{\today}
\maketitle


\section{Introduction}\label{sec:intro}
The Hardy inequality 
\begin{equation}\label{eq:hardy0}
  \int_{\R^n}\frac{|\phi|^2}{|x|^2}\,dx
  \leq
  \left(\frac{2}{n-2}\right)^2\int_{\R^n}|\nabla\phi|^2\,dx
  \qquad
  (n\geq3)
\end{equation}
with $\psi\in\mathcal C^{\infty}_0(\R^n;\C)$,
is one of the well known mathematical aspects of the uncertainty principle in Quantum Mechanics.
Among several applications, a standard consequence of \eqref{eq:hardy0} is the positivity of quadratic forms of the type
\begin{equation*}
  q(\varphi,\psi)=\int_{\R^n}\nabla\varphi\cdot\nabla\overline{\psi}\,dx
  -\lambda\int_{\R^n}\frac{\varphi\overline{\psi}}{|x|^2}\,dx
\end{equation*}
for $\lambda\leq(n-2)^2/4$,
which permits to study the self-adjointness of Schr\"odinger hamiltonians like $H=-\Delta-\lambda/|x|^2$ by means of the Kato-Rellich Theorem.

One can prove by the same techniques the more general family of inequalities
\begin{equation}\label{eq:hardy1}
  \int_{\R^n}\frac{|\phi|^2}{|x|^\alpha}\,dx
  \leq
  \left(\frac{2}{n-\alpha}\right)^2\int_{\R^n}|x|^{2-\alpha}|\nabla\phi|^2\,dx
  \qquad
  (n\geq1)
  \qquad
  (\alpha<n)
\end{equation}
In fact, the operator $\nabla$ at the right-hand side of \eqref{eq:hardy1} can be replaced by the radial derivative $\partial_r=\frac{x}{|x|}\cdot\nabla$, since the weight $|x|^{-\alpha}$ is radial; more precisely, one has 
\begin{equation}\label{eq:hardy2}
  \int_{\R^n}\frac{|\phi|^2}{|x|^\alpha}\,dx
  \leq
  \left(\frac{2}{n-\alpha}\right)^2\int_{\R^n}|x|^{2-\alpha}|\partial_r\phi|^2\,dx
  \qquad
  (n\geq1)
  \qquad
  (\alpha<n)
\end{equation}
In addition, the constant $4/(n-\alpha)^2$ at the right-hand side of \eqref{eq:hardy1}, \eqref{eq:hardy2} is sharp, and it is well known that there are no maximizing functions for those inequalities.

When a particle interacts with an external magnetic field, it is standard in Quantum Mechanics to introduce in the model an anti-symmetric real-valued matrix $B=\{B^{jk}(x)\}=\{-B^{kj}(x)\}:\R^n\to\mathcal M_{n\times n}(\R)$, $j,k=1,\dots,n$, $n\geq2$, with the following property: there exist a real valued potential vector field $A=\left(A^1(x),\dots,A^n(x)\right)$ such that $B=DA-DA^t$, where $(DA)^{ij}=\partial_{x_i}A^j$ is the differential matrix of $A$. Then, to obtain the new Schr\"odinger hamiltonian formulation, one changes the gradient $\nabla$ into $\nabla_A:=\nabla+iA$, so that $-\Delta_A=-(\nabla_A)^2$.
In dimensions $n=2,3$, any anti-symmetric matrix can be identified with a scalar function $(n=2)$ or a vector-field $(n=3)$, hence we have, by the previous definitions, that $B=\text{curl}\,A$.

A quite important feature of the magnetic gradient $\nabla_A$ is the diamagnetic inequality (see e.g. \cite{LL}): if $A\in_{L^2_{\text{loc}}}(\R^n)$, $n\geq2$, then 
\begin{equation}\label{eq:diam}
  \left|\nabla|\psi(x)|\right|\leq\left|\nabla_A\psi(x)\right|
\end{equation}
for all $\psi\in\mathcal C^\infty_0(\R^n)$ and almost every $x\in\R^n$. This, together with \eqref{eq:hardy1} applied to $|\psi|$, immediately gives 
\begin{equation}\label{eq:hardy3}
  \int_{\R^n}\frac{|\phi|^2}{|x|^\alpha}\,dx
  \leq
  \left(\frac{2}{n-\alpha}\right)^2\int_{\R^n}|x|^{2-\alpha}|\nabla_A\phi|^2\,dx
  \qquad
  (n\geq2)
  \qquad
  (\alpha<n)
\end{equation}
where the constant at the right-hand side is not bigger than the one in the free case. 
It is easy to show again that one can put at the right-hand side of \eqref{eq:hardy3} just the radial component of the magnetic gradient, namely
\begin{equation}\label{eq:hardy33}
  \int_{\R^n}\frac{|\phi|^2}{|x|^\alpha}\,dx
  \leq
  \left(\frac{2}{n-\alpha}\right)^2\int_{\R^n}|x|^{2-\alpha}|\partial_r^A\phi|^2\,dx
  \qquad
  (n\geq2)
  \qquad
  (\alpha<n)
\end{equation}
holds,
where $\partial_r^A:=\frac{x}{|x|}\cdot\nabla_A$.
 
In \cite{BLS}, \cite{LW} it is proved that, for suitable magnetic fields $B$, 
inequality \eqref{eq:hardy3} can be generally strictly improved with respect to the free case. 
A relevant example is given by the Aharonov-Bohm potential in 2D: it is the case
\begin{equation*}
  A_{ab}(x,y)=\lambda\left(-\frac{y}{x^2+y^2},\frac{x}{x^2+y^2}\right),
  \qquad
  B_{ab}(x,y)=\text{curl}\,A=8\pi\delta
\end{equation*}
with $\lambda\in\R$. Laptev and Weidl proved in \cite{LW} that
\begin{equation}\label{eq:hardyLW}
  \int_{\R^2}\frac{|\phi|^2}{|x|^2}\,dx\leq
  \Gamma\int_{\R^2}|\nabla_{A_{ab}}\phi|^2\,dx
\end{equation}
with $\Gamma=\left(\text{dist}\{\Theta,\Z\}\right)^{-2}$, being $\Theta$ the total flux of $A_{ab}$ along the unit sphere $\mathbb S^1$. Notice that inequality \eqref{eq:hardyLW} is false in the free case, since in dimension $n=2$ the weight $|x|^{-2}$ is too singular. Nevertheless, as soon as $\Theta\notin\Z$, \eqref{eq:hardyLW} becomes true with a finite constant $\Gamma$ at the right-hand side.

This manuscript is concerned with the same kind of questions in the relativistic setting of Dirac-Pauli operators.
We denote by
\begin{equation*}
  \sigma\cdot\nabla_A \phi= \sum_{j=1}^3 \sigma_j (\partial_{x_j} +i A_j)\phi
\end{equation*}
where $\phi=\phi(x)=(\phi^1(x),\phi^2(x)):\R^3 \rightarrow \C^2$ and the Pauli matrices $\sigma_j$ are defined by
\begin{equation}\label{eq:sigma}
  \sigma_1=
  \left(
  \begin{array}{cc}
  0 & 1 \\ 1 & 0
  \end{array}
  \right),
  \qquad
  \sigma_2=
  \left(
  \begin{array}{cc}
  0 & -i \\ i & 0
  \end{array}
  \right),
  \qquad
  \sigma_3=
  \left(
  \begin{array}{cc}
  1 & 0 \\ 0 & -1
  \end{array}
  \right)
\end{equation}
The anti commutation relations
\begin{equation*}
  \sigma_j\sigma_k+\sigma_k\sigma_j=
  \begin{cases}
 2I_{2\times2}
  \qquad
  \text{if }j=k
  \\
  \mathbf{0}
  \qquad\quad\ \ \ 
  \text{if } j\neq k
  \end{cases},
  \qquad
  I_{2\times2}=
  \left(
  \begin{array}{cc}
  1 & 0 \\ 0 & 1
  \end{array}
  \right)
\end{equation*}
are the main feature of the matrices $\sigma_j$; in particular, they give
\begin{equation}\label{eq:ovvia}
  \int_{\R^3}\left|\sigma\cdot\nabla\phi\right|^2\,dx=
  \int_{\R^3}|\nabla\phi|^2\,dx=:
  \int_{\R^3}|\nabla\phi^1|^2\,dx
  +\int_{\R^3}|\nabla\phi^2|^2\,dx
\end{equation}
(see also \cite{T} for more details)
so that by \eqref{eq:hardy0} one immediately obtains
\begin{equation}\label{eq:hardy5}
\int_{\R^3}\frac{|\phi|^2}{|x|^2}\,dx:=
\int_{\R^3}\frac{|\phi^1|^2}{|x|^2}\,dx
+\int_{\R^3}\frac{|\phi^2|^2}{|x|^2}\,dx
\leq
4 \int_{\R^3}\left|\sigma\cdot \nabla \phi\right|^2\,dx
\end{equation}
Notice that 
\begin{equation*}
  \int_{\R^3}|x|^\alpha\left|\sigma\cdot\nabla\phi\right|^2\,dx
  \neq
  \int_{\R^3}|x|^\alpha|\nabla\phi|^2\,dx
\end{equation*}
if $\alpha\neq0$, so that obtaining the relativistic analog of \eqref{eq:hardy1} is not trivial.

The case $\alpha=1$ is of particular interest, in connection with the problem of self-adjointness of Dirac operators with an external Coulomb-type potential. It is the case of the following inequality
\begin{equation}\label{eq:hardy6}
   \int_{\R^3}\frac{|\phi|^2}{|x|}\,dx\leq \int_{\R^3}|x|\left|\sigma\cdot\nabla\phi\right|^2\,dx
\end{equation}
From now on, we shall refer to \eqref{eq:hardy6}
as to the Hardy-Dirac type inequality. 

In \cite{DELV},
a completely analytical proof of estimate \eqref{eq:hardy6} has been 
performed for the first time, and later by the same techniques some more general versions have been obtained in \cite{A, DDEV, DEL, DEL2} (see also the references therein). Notice that the best constant of the inequality
is $C=1$, in complete analogy with \eqref{eq:hardy1}, for $\alpha=1$ and $n=3$. 

The aim of this paper is to investigate the validity of the following inequalities
\begin{equation}\label{eq:hardy}
   \int_{\R^3}\frac{|\phi|^2}{|x|}\,dx\leq C \int_{\R^3}|x|\left|\sigma\cdot\nabla_A\phi\right|^2\,
dx
\end{equation}
when $A:\R^3\to\R^3$ is a suitable homogeneous magnetic potential, with a particular interest in studying the behavior of the best constant $C$. 

Before preparing the setting of our main results, we motivate here the interest for such a question. First of all,
we recall the following well known Barry Simon's version of the diamagnetic inequality:
\begin{equation}\label{eq:diamB}
  e^{-t(\Delta+V(x))}\left|f\right|\leq
  \left|e^{-t(\Delta_A+V(x))}f\right|
\end{equation}
if $A\in L^2_{\text{loc}}(\R^n)$, $n\geq2$ which was first conjectured in \cite{S1}. 
Inequality \eqref{eq:diamB} implies, among many other things, that the bottom of the spectrum of an electromagnetic Schr\"odinger operator $-\Delta_A+V$ increases with respect to the same quantity in the magnetic-free case, namely
\begin{equation}\label{eq:unidia}
  \inf\text{spec}(-\Delta_A+V)\geq\inf\text{spec}(-\Delta+V)
\end{equation} 
in $L^2(\R^n)$, $n\geq2$. The same phenomenon (universal diamagnetism) does not seem to arise in relativistic models. 
More precisely, in \cite{HSS}, an interesting conjecture about universal paramagnetism for fermions was claimed, which for Dirac-Pauli operators can be written as
\begin{equation}\label{eq:unipara}
 \inf\text{spec}(-\Delta_A+\sigma\cdot B+V)\leq\inf\text{spec}(-\Delta+V),
\end{equation} 
with $n=3$ and $B=\text{curl}\,A$.
Moreover, in \cite{HSS} it is proved that the claim is true in the case of a constant magnetic field $B = (0,0,\lambda)$, $\lambda\in\R$. Later on, Avron and Simon in \cite{AS} disproved the conjecture with an explicit example. See also the interesting surveys \cite{E, H} for more informations about the topic.

The flavour is that inequalities \eqref{eq:unidia}, \eqref{eq:unipara} should be directly related to the behavior of the best constant in the Hardy inequalities \eqref{eq:hardy3}, \eqref{eq:hardy}, respectively, once a magnetic perturbation comes into play.

Motivated by the result by Avron and Simon, we wish here to show a general class of non-trivial magnetic fields for which the best constant in \eqref{eq:hardy} is the same as in the free case, namely $C=1$, in contrast with the paramagnetic phenomenon. The argument which we show in the sequel does not permit us to find example of fields for which \eqref{eq:hardy} holds with a better constant, which is a quite interesting open question.

Before stating our main result we need to introduce the {\it orbital angular momentum} operator, which is the triplet of operators
\begin{equation}\label{eq:1momfree}L = (L^1, L^2, L^3)
 =x\wedge(-i\nabla)
 \end{equation}
and the {\it magnetic orbital angular momentum}
\begin{equation}\label{eq:1mom}L_A = (L_A^1, L^2_A, L^3_A)
 =x\wedge(-i\nabla_A)=x\wedge(-i\nabla+A)
 \end{equation}
Notice that the operators $L,L_A$ 
acts in principle on $\C^2$-valued functions defined on $\R^3$. However
it is also well-defined its action on $\C^2$-valued functions defined on $\mathbb S^2$ (consider the homogeneous extension
on the whole $\R^3$ of a given function defined on $\mathbb S^2$, apply the operator and restrict on $\mathbb S^2$).
The corresponding operators depend only on the 
trace of the field $A$ on $\mathbb S^2$. Hence in the sequel, given any potential $A:\R^3\to\R^3$, we will denote by $L_A$ both the symmetric operator on $L^2(\mathbb S^2; \mathbb C^2)$
and the one on $L^2(\R^3; \mathbb C^2)$. 

A fundamental role is played by the
{\it spin-orbit} angular momentum $\sigma\cdot L+1$; for any
vector field $A$, its natural generalization is given by 
$\sigma \cdot L_{A}+1$. Recall that $\sigma_i$ are defined in \eqref{eq:sigma}.\\
The operator $\sigma\cdot L_{ A}$ is symmetric and its inverse on $L^2(\mathbb S^2;\C^2)$ is compact. 
Hence it has a purely discrete and real spectrum, which can accumulate only at infinity; we denote it by
\begin{equation}\label{eq:spec}
\begin{cases}
\text{spec}\left(\sigma\cdot {L}_{A}\right):=\{-\lambda_j,\mu_j\}_{j\in \N}
\\
0<\lambda_1<\lambda_2<\dots;
\qquad
0\leq\mu_1<\mu_2<\dots
\end{cases}
\end{equation}
In particular it is well defined the number 
\begin{equation}\label{eq:mu1}
\mu_1(A)=\inf \left\{\text{spec} (\sigma \cdot{L}_{A}+1) \cap [0, \infty)\right\}
\end{equation}
In the magnetic-free case ${A}\equiv0$, the spectrum of $\sigma\cdot L+1$ is completely known, and it is given by the set $\{\pm1,\pm2,\dots\}$. This gives $\mu_1(0)=1$, $\mu_1$ being the number in \eqref{eq:mu1}, which turns out to be the fundamental tool in the proof of \eqref{eq:hardy6} of \cite{DELV}.
In fact, thanks to this remark, in \cite{DELV} the stronger estimate
\begin{equation}\label{eq:stronger}
\int_{\R^3}\frac{|\phi|^2}{|x|}\,dx
   \leq
   \int_{\R^3} \frac1{|x|}\left|\left(\sigma\cdot L+1\right)\phi\right|^2 \,dx
   \leq
   \int_{\R^3}|x|\left|\sigma\cdot\nabla\phi\right|^2\,dx
\end{equation}
is proved, providing a weighted $L^2$-bound for the spin-orbit angular momentum in terms of the whole Dirac operator.

In the following, we use the polar notations $r=|x|,\ \omega=x/|x|\in\mathbb S^2$.
We will point our attention on magnetic fields of the form
\begin{equation*}
    B(x)=\varphi(r) \nabla g(\omega)\wedge x,
\end{equation*}
where $\varphi=\varphi(r):\R^3\to\R$, and $g=g(\omega):\mathbb S^2\to\R$ is a homogenous function of degree $0$. These kinds of fields are obviously tangential, and possibly singular at the origin, since $\nabla g$ is a homogeneous function of degree $-1$. As we will see in the sequel (Proposition \ref{prop:gauge} below), up to assuming some local integrability conditions on $\varphi$ (assumption \eqref{eq:sing}), it is a possible to prove that there exists a potential $A$ such that $\text{curl}\,A=B$.
In particular, we will prove that for those potentials one has $\text{spec} (\sigma \cdot{L}_{A}+1)
=\text{spec} (\sigma \cdot{L}+1)$, and the corresponding eigenfunctions are just obtained by the free ones, via multiplication by a purely imaginary phase (Proposition \ref{prop:gauge2} below).

We are now ready to state the main result of this paper.
\begin{thm}\label{thm:dir2}
  Let $\varphi=\varphi(r):\R^3\to\R$, $g=g(\omega):\mathbb S^2\to\R$ be a homogenous function of degree $0$, and denote by
  \begin{equation}\label{eq:condizione}
    B(x)=\varphi(r) \nabla g(\omega)\wedge x.
  \end{equation} 
  Assume in addition that
  \begin{equation}\label{eq:sing}
    \left|\int_0^rs\varphi(s)\,ds\right|<\infty,
  \end{equation}
  for all $r\in\R$. Moreover, let $A:\R^3\to\R^3$ be such that $\text{curl}\,A=B$.
  Then, for any $\phi=\phi(x)\in\mathcal C^{\infty}_0(\R^3;\C^2)$, the following inequality holds
\begin{equation}\label{eq:hardy7}
   \int_{\R^3}\frac{|\phi|^2}{|x|}\,dx
   \leq
   \int_{\R^3} \frac1{|x|}\left|\left(\sigma\cdot L_A+1\right)\phi\right|^2 \,dx
   \leq
   \int_{\R^3}|x|\left|\sigma\cdot\nabla_A\phi\right|^2\,dx.
\end{equation}
\end{thm}
\begin{remark}
  Notice that condition \eqref{eq:sing} does not allow to consider too much singular magnetic fields $B$. Since $\nabla g(\omega)\wedge x$ is homogeneous of degree 0, the validity of \eqref{eq:sing} requires on $B$ a local behavior like $|B(x)|\sim1/|x|^{2-\epsilon}$, for some $\epsilon>0$. On the other hand, nothing is required about the behavior of $B$ when $r\to\infty$.
\end{remark}
\begin{remark}
  The result of Theorem \ref{thm:dir2} is gauge invariant, since the hypotheses and inequality \eqref{eq:hardy7} do not depend on the choice of the potential $A$ such that $\text{curl}\,A=B$. As we see in the following, in order to prove \eqref{eq:hardy7} it is fundamental to choose an appropriate gauge for the potential $A$. Moreover we remark again that condition \eqref{eq:sing} permits to prove that $B$ is in fact a curl of some potential $A$, which will be chosen in a suitable gauge (see Proposition \ref{prop:gauge2} below).
\end{remark}
\begin{remark}
  Theorem \ref{thm:dir2} shows a class of non-trivial magnetic fields $B$, with the corresponding potentials $A$,
for which the best constant 1 of inequality \eqref{eq:hardy7} coincides with the one in the free case \eqref{eq:hardy6}, in the same spirit as in the example by Avron and Simon \cite{AS} remarked above.
\end{remark}
\begin{remark}\label{rem:examples}
  A quite interesting example of magnetic field for which Theorem \ref{thm:dir2} applies is given by
  \begin{equation}\label{eq:exampleB}
    B(x,y,z)=\lambda r^{\alpha}(-y,x,0)
    =\lambda r^{\alpha+1}\left[\nabla\left(\frac zr\right)\wedge x\right],
  \end{equation}
  with $\lambda\neq0$ and $r:=\left(x^2+y^2+z^2\right)^{\frac12}$. Following the notations in \eqref{eq:condizione} we have $\varphi(r)=\lambda r^{\alpha+1}$ and $g(\omega)=z/r$; moreover, condition \eqref{eq:sing} impose that $\alpha>-3$.
\end{remark}
\begin{remark}
The question whether or not there exist a magnetic field $B$ 
for which inequality \eqref{eq:hardy7}, with the corresponding potential, holds with a constant better than 1, in analogy with what happens in the non relativistic case \eqref{eq:hardyLW}, still remains open. In order to address an answer it should be fundamental to produce examples of magnetic potentials such that $\mu_1(A)>1$ in the definition \eqref{eq:mu1}.
\end{remark}


\section{Preliminaries}\label{sec:prelim}
We start with some preliminary remarks which will be used in the sequel, in the proof of our main Theorem \ref{thm:dir2}. We first need to fix a suitable gauge in which to work. We prove the following proposition.
\begin{prop}[Gauge choice]\label{prop:gauge}
 Let $B=B(x):\R^3\to\R^3$ be of the form \eqref{eq:condizione} and assume \eqref{eq:sing}. Define $A=A(x):\R^3\to\R^3$ as follows:
 \begin{equation}\label{eq:gauge}
   A(x) = \frac12\varphi(r)g(\omega)x-\frac12\left(\int_0^rs\varphi(s)\,ds\right)\nabla g(\omega).
 \end{equation}
 Then
 \begin{align}
 &
 \text{curl}\, A = B 
 \label{eq:curla}
 \\
 &
 \partial_r\left(x\wedge A\right)+x\wedge\nabla\left(A\cdot \frac xr\right)=0,
 \label{eq:gauge2}
 \end{align}
 where the radial derivative $\partial_r:=\frac xr\cdot\nabla$ acts component-wise on the vector $x\wedge A$.
\end{prop}
\begin{proof}
  The proof is quite elementary. First, compute
  \begin{align*}
    \text{curl}\,\varphi(r)g(\omega)x 
    &
    = \varphi(r)\nabla g(\omega)\wedge x = B(x)
    \\
    \text{curl}\,\left[\left(\int_0^rs\varphi(s)\,ds\right)\nabla g(\omega)\right]
    &
    =r\varphi(r)\frac xr\wedge\nabla g(\omega)
    =-B(x),
  \end{align*}
  which proves \eqref{eq:curla}. Now notice that
  \begin{equation*}
    x\wedge A = -\frac12\left(\int_0^rs\varphi(s)\,ds\right)x\wedge\nabla g(\omega);
  \end{equation*}
  since $x\wedge\nabla g(\omega)$ is homogeneous of degree 0, we have $\partial_r(x\wedge\nabla g(\omega))\equiv0$ and consequently
  \begin{equation}\label{eq:00}
    \partial_r\left(x\wedge A\right)
    =-\frac12 r\varphi(r)x\wedge\nabla g(\omega)=\frac12rB(x).
  \end{equation}
  Then, compute $A\cdot \frac xr=\frac12\varphi(r)g(\omega)r$, to obtain
  \begin{equation}\label{eq:000}
     x\wedge\nabla\left(A\cdot \frac xr\right)
     =
     \frac12r\varphi(r)x\wedge\nabla g(\omega)=-\frac12rB(x).
  \end{equation}
  Identities \eqref{eq:00} and \eqref{eq:000} complete the proof of \eqref{eq:gauge2}.
\end{proof}
We now need a further simple remark about the spectral properties of $\sigma\cdot L_A+1$.
\begin{prop}\label{prop:gauge2}
  Let $A=A(x):\R^3\to\R^3$ be defined as in \eqref{eq:gauge}, assuming \eqref{eq:sing}.
  Moreover, denote by 
  \begin{equation*}
    \eta(x):=\frac12\left(\int_0^rs\varphi(s)\,ds\right)g(\omega).
   \end{equation*} 
Then,  $L_A\left(e^{i\eta}\phi\right)=e^{i\eta}L\phi$, where $L,L_A$ are defined by \eqref{eq:1momfree} and \eqref{eq:1mom}. In particular, the spectra of $\sigma\cdot L_{A}+1$ and $\sigma\cdot
L_{A}+1$ on $L^2(\mathbb S^2;\C^2)$ coincide, and consequently
  \begin{equation}\label{eq:sigmafin}
    \left\|\left(\sigma\cdot L_{ A}+1\right)\phi\right\|_{L^2(\mathbb S^2;\C^2)}
    \geq\|\phi\|_{L^2(\mathbb S^2;\C^2)}.
  \end{equation}
\end{prop}
\begin{proof}
  The proof is quite immediate. The identity  $L_A\left(e^{i\eta}\phi\right)=e^{i\eta}L\phi$ can be easily checked via explicit computations. Moreover, \eqref{eq:sigmafin} immediately follows by its analog in the free case $A\equiv0$ (see e.g. \cite{T}) and the fact that the the spectra of $\sigma\cdot L_{A}+1$ and $\sigma\cdot
L_{ A}+1$ on $L^2(\mathbb S^2; \mathbb C^2)$ coincide.
\end{proof}
We now have all the ingredients which we can use to prove Theorem \ref{thm:dir2}.

 
\section{Hardy-Dirac Inequalitites: proof of Theorem \ref{thm:dir2}}\label{sec:hardy}
Let us start by proving an identity.
\begin{prop}\label{prop:identity}
Let $A:\R^3\to\R^3$. 
For any $\phi=\phi(x)\in\mathcal C^{\infty}_0(\R^3;\C^2)$, the following identity holds
\begin{align}\label{eq:identity}
\int_{\R^3} r\left|\sigma\cdot\nabla_A\phi\right|^2
&
=
\int_{\R^3} r\left|\partial_r^A\phi\right|^2 \,dx\\
\nonumber &
+\int_{\R^3} r\left|\frac1{r}\left(\sigma\cdot L_A+1\right)\phi\right|^2 \,dx
-\int_{\R^3} \frac{|\phi|^2}{r} \,dx
\\
\nonumber
& 
+\int_{\R^3}\langle\sigma\cdot[\partial_r(x\wedge A)]\phi,\phi\rangle \,dx
+\int_{\R^3}\left\langle\sigma\cdot\left(x\wedge\nabla A_r\right)\phi,\phi\right\rangle \,dx
\end{align}
where $r:=|x|$, $\partial_r^A:=\frac xr\cdot\nabla_A$, and $A_r:=A\cdot \frac xr$.
\end{prop}
\begin{proof}
Let us recall the orthogonal decompositions
\begin{align*}
  \nabla 
  &
  = \frac xr\partial_r-\frac xr\wedge \left(\frac xr\wedge \nabla\right)
  \\
  iA
  &
  = i\frac xrA_r-i\frac xr\wedge\left(\frac xr\wedge A\right),
\end{align*}
where $\partial_r:=\frac xr\cdot\nabla$.
Then, denoting by $\partial_r^A=\frac xr\cdot\nabla_A=\partial_r+iA_r$, we can write
\begin{equation}\label{eq:dec1}
  \nabla_A = \nabla+iA = \frac xr\partial_r^A-\frac xr\wedge\left(\frac xr\wedge \nabla_A\right)
\end{equation}
Notice that, since $x/r$ and $x/r\wedge \nabla_A$ are orthogonal, we have by \eqref{eq:dec1} that
\begin{equation}\label{eq:pitagora}
   |\nabla_A\phi|^2 = |\partial_r^A\phi|^2+\left|\frac xr\wedge \nabla_A\phi\right|^2
   =
   |\partial_r^A\phi|^2+\frac1{r^2}\left|L_A\phi\right|^2.
\end{equation}
We recall the anti commutation rules
\begin{equation}\label{eq:3mom}
  \sigma_2\sigma_3=i\sigma_1,
  \qquad
  \sigma_3\sigma_1=i\sigma_2,
  \qquad
  \sigma_1\sigma_2=i\sigma_3,
  \qquad
  \sigma_j\sigma_k+\sigma_k\sigma_j=2\delta_j^kI,
\end{equation}
and the useful formula
\begin{equation}\label{eq:formula}
(\sigma\cdot F)(\sigma\cdot G) = F\cdot G+i\sigma\cdot(F\wedge G),
\end{equation}
which follows by \eqref{eq:3mom}.
By \eqref{eq:1mom}, \eqref{eq:formula} and \eqref{eq:dec1} one easily obtains that
\begin{equation}\label{eq:dec2}
   \sigma\cdot\nabla_A
   =
   \left(\sigma\cdot \frac xr\right)
   \left(\partial_r^A-\frac1r\sigma\cdot L_{A}\right)
\end{equation}
In addition, due to the anti-commutation rules \eqref{eq:3mom}
 one has
$\left|\left(\sigma\cdot \frac xr\right)F\right|^2=|F|^2$, for any vector $F\in\C^2$. Hence, we can compute by \eqref{eq:dec2}:
\begin{align}\label{eq:parti1}
&\int_{\R^3} r\left|\sigma\cdot\nabla_A\phi\right|^2 \,dx
=\int_{\R^3} r\left|\partial_r^A\phi-\frac 1r\sigma\cdot L_{A}\phi\right|^2 \,dx
\\
\nonumber
& 
=
\int_{\R^3} r\left|\partial_r\phi+iA_r\phi-\frac 1r\sigma\cdot L\phi-\sigma\cdot\left(\frac xr\wedge A\right)\phi\right|^2 \,dx
\\
\nonumber
&
=\int_{\R^3} r\left|\partial_r^A\phi\right|^2 \,dx
+\int_{\R^3}\frac1{r}\left|\sigma\cdot L_A\phi\right|^2 \,dx\\
\nonumber
&
-\int_{\R^3} \left\langle\partial_r\phi,\sigma\cdot L\phi\right\rangle \,dx
-\int_{\R^3} \left\langle\sigma\cdot L\phi,\partial_r\phi\right\rangle \,dx
\\
\nonumber
& \ \ \ 
-\int_{\R^3} \left\langle\partial_r\phi,\sigma\cdot (x\wedge A)\phi\right\rangle \,dx
-\int_{\R^3} \left\langle\sigma\cdot (x\wedge A)\phi,\partial_r\phi\right\rangle \,dx
\\
\nonumber
& \ \ \ 
-\int_{\R^3}\left\langle iA_r\phi,\sigma\cdot L\phi\right\rangle \,dx
-\int_{\R^3}\left\langle\sigma\cdot L\phi,iA_r\phi\right\rangle \,dx
\end{align}
where the brackets $\langle\cdot,\cdot\rangle$ denote the sesquilinear product in $\C^2$, and we used the fact that $\langle iA_r\phi,\sigma\cdot(x\wedge A)\phi\rangle+\langle\sigma\cdot(x\wedge A)\phi\
,iA_r\phi \rangle=0$.
Now notice that $[\partial_r,\sigma\cdot L]=0$ and $\sigma\cdot L$ is a symmetric operator; therefore, integrating by parts we obtain
\begin{align}\label{eq:a}
&-\int_{\R^3} \left\langle\partial_r\phi,\sigma\cdot L\phi\right\rangle \,dx
-\int_{\R^3} \left\langle\sigma\cdot L\phi,\partial_r\phi\right\rangle \,dx
=
\\
\nonumber
&
-\int_{\R^3}\partial_r\langle\sigma\cdot L\phi,\phi\rangle \,dx
=
\int_{\R^3}\frac2r\langle\sigma\cdot L\phi,\phi\rangle \,dx
\end{align}
Analogously, we can write
\begin{align*}
&
\int_{\R^3} \left\langle\partial_r\phi,\sigma\cdot (x\wedge A)\phi\right\rangle \,dx
+\int_{\R^3} \left\langle\sigma\cdot (x\wedge A)\phi,\partial_r\phi\right\rangle \,dx
\\
& 
=
\int_{\R^3}\partial_r\langle\sigma\cdot(x\wedge A)\phi,\phi\rangle \,dx
-\int_{\R^3}\langle\sigma\cdot[\partial_r(x\wedge A)]\phi,\phi\rangle \,dx
\end{align*}
where the radial derivative $\partial_r$ acts component-wise on the vector $x\wedge A$. Consequently, by integration by parts we get
\begin{align}
\label{eq:b}
&
-\int_{\R^3} \left\langle\partial_r\phi,\sigma\cdot (x\wedge A)\phi\right\rangle \,dx
-\int_{\R^3} \left\langle\sigma\cdot (x\wedge A)\phi,\partial_r\phi\right\rangle \,dx
\\
\nonumber
& 
=
\int_{\R^3}\frac 2r\langle\sigma\cdot(x\wedge A)\phi,\phi\rangle \,dx
+\int_{\R^3} \langle\sigma\cdot[\partial_r(x\wedge A)]\phi,\phi\rangle \,dx
\end{align}
Hence the sum of \eqref{eq:a} and \eqref{eq:b} gives
\begin{align}
\label{eq:c}
&
-\int_{\R^3} \left\langle\partial_r\phi,\sigma\cdot L\phi\right\rangle \,dx
-\int_{\R^3} \left\langle\sigma\cdot L\phi,\partial_r\phi\right\rangle \,dx
\\
\nonumber
&
-\int_{\R^3} \left\langle\partial_r\phi,\sigma\cdot (x\wedge A)\phi\right\rangle \,dx
-\int_{\R^3} \left\langle\sigma\cdot (x\wedge A)\phi,\partial_r\phi\right\rangle \,dx
\\
\nonumber
& 
=
\int_{\R^3} \frac 2r\langle\sigma\cdot L_A\phi,\phi\rangle \,dx
+\int_{\R^3}\langle\sigma\cdot[\partial_r(x\wedge A)]\phi,\phi\rangle \,dx
\end{align}
For the last two terms in \eqref{eq:parti1}, since $\sigma\cdot L$ is symmetric and $A_r$ commutes with the $\sigma_j's$ we easily compute
\begin{align}\label{eq:d}
&\int_{\R^3}\left\langle iA_r\phi,\sigma\cdot L\phi\right\rangle \,dx
+\int_{\R^3}\left\langle\sigma\cdot L\phi, iA_r\phi\right\rangle \,dx
\\
\nonumber
&
=
i\int_{\R^3} \left\langle\sigma\cdot L(A_r\phi),\phi\right\rangle \,dx
-i\int_{\R^3} \left\langle A_r\sigma\cdot L\phi,\phi\right\rangle \,dx
\\
\nonumber
&
=
\int_{\R^3} \left\langle\sigma\cdot\left(x\wedge\nabla A_r\right)\phi,\phi\right\rangle \,dx
\end{align}
Combining \eqref{eq:parti1}, \eqref{eq:c} and \eqref{eq:d} it now turns out that
\begin{align*}
\int_{\R^3} r\left|\sigma\cdot\nabla_A\phi\right|^2 \,dx
&
=
\int_{\R^3} r\left|\partial_r^A\phi\right|^2 \,dx
+\int_{\R^3}\frac1{r}\left|\sigma\cdot L_A\phi\right|^2 \,dx
+\int_{\R^3}\frac 2r\langle\sigma\cdot L_A\phi,\phi\rangle \,dx
\\
\nonumber
& \ \ \ 
+\int_{\R^3}\langle\sigma\cdot[\partial_r(x\wedge A)]\phi,\phi\rangle \,dx
+\int_{\R^3}\left\langle\sigma\cdot\left(x\wedge\nabla A_r\right)\phi,\phi\right\rangle \,dx
\\
\nonumber
& 
=
\int_{\R^3} r\left|\partial_r^A\phi\right|^2\,dx
+\int_{\R^3}\frac1{r}\left|\left(\sigma\cdot L_A+1\right)\phi\right|^2 \,dx
-\int_{\R^3}\frac{|\phi|^2}{r} \,dx
\\
\nonumber
& \ \ \ 
+\int_{\R^3}\langle\sigma\cdot[\partial_r(x\wedge A)]\phi,\phi\rangle \,dx
+\int_{\R^3}\left\langle\sigma\cdot\left(x\wedge\nabla A_r\right)\phi,\phi\right\rangle \,dx
\end{align*}
which proves \eqref{eq:identity}.
\end{proof}
Theorem \ref{thm:dir2} now follows as an immediate corollary of Proposition \ref{prop:identity}, thanks to the Hardy inequality
  \begin{equation*}
     \int_{\R^3}\frac{|\phi|^2}{|x|}\,dx\leq\int_{\R^3}|x|\left|\partial_A^r\phi\right|^2\,dx
  \end{equation*}
  (which is \eqref{eq:hardy33} with $\alpha=1$, $n=3$), and Propositions \ref{prop:gauge}, \ref{prop:gauge2}.

\end{document}